%% file: ISITmain.tex
\documentclass[10pt, conference, letterpaper]{IEEEtran}

\IEEEoverridecommandlockouts

\usepackage{color}
\usepackage[dvipsnames]{xcolor}
\usepackage{balance}

\usepackage{amsthm}
\usepackage{mathtools}

\usepackage{bbm}
\usepackage{amssymb}
\usepackage{tikz}

\usepackage{enumitem}
\usepackage[algo2e,ruled,vlined,linesnumbered]{algorithm2e}
\SetInd{0em}{0.5em}
\SetKwComment{Comment}{/\!\!/}{}
\usepackage{booktabs}
\usepackage{cite}
\usepackage{caption}
\usepackage{subfig}

\newcommand\blfootnote[1]{%
  \begingroup
  \renewcommand\thefootnote{}\footnote{#1}%
  \addtocounter{footnote}{-1}%
  \endgroup
}

\newtheorem{Theorem}{Theorem}

\theoremstyle{definition}

\newtheorem{Definition}{Definition}
\newtheorem{Example}{Example}

\newtheorem{ Proof}{Proof}

\input{commands.tex}

\begin{document}

\title{ \fontsize{23.9pt}{30pt}\selectfont On Single-User Interactive Beam Alignment in Millimeter Wave Systems: Impact of Feedback Delay}

\author{Abbas Khalili$^\dagger$, Shahram Shahsavari$^\ddagger$, Mohammad A. (Amir) Khojastepour$^\diamond$, Elza Erkip$^\dagger$. \\
$^\dagger$NYU Tandon School of Engineering, $^\ddagger$University of Waterloo, $^\diamond$NEC Laboratories America, Inc.. \\
Emails: $^\dagger$\{ako274, Elza\}@nyu.edu, $^\ddagger$shahram.shahsavari@uwaterloo.ca,  $^\diamond$amir@nec-labs.com.}

\maketitle
\begin{abstract}

Narrow beams are key to wireless communications in millimeter wave frequency bands. Beam alignment (BA) allows the base station (BS) to adjust the direction and width of the beam used for communication. During BA, the BS transmits a number of scanning beams covering different angular regions. The goal is to minimize the expected width of the uncertainty region (UR) that includes the angle of departure of the user. Conventionally, in interactive BA, it is assumed that the feedback corresponding to each scanning packet is received prior to transmission of the next one. However, in practice, the feedback delay could be larger because of propagation or system constraints. This paper investigates BA strategies that operate under arbitrary fixed feedback delays. This problem is analyzed through a source coding prospective where the feedback sequences are viewed as source codewords. It is shown that these codewords form a codebook with a particular characteristic which is used to define a new class of codes called \textit{$d-$unimodal} codes. By analyzing the properties of these codes, a lower bound on the minimum achievable expected beamwidth is provided. 
The results reveal potential performance improvements in terms of the BA duration it takes to achieve a fixed expected width of the UR over the state-of-the-art BA methods which do not consider the effect of delay. 
\blfootnote{This work is supported by
National Science Foundation grants EARS-1547332, 
SpecEES-1824434, and NYU WIRELESS Industrial Affiliates.}
\end{abstract}

\begin{IEEEkeywords}
Millimeter wave, Analog beam alignment, Interactive beam alignment, Non-interactive beam alignment, Contiguous beams.
\end{IEEEkeywords}

\section{Introduction}

Millimeter wave (mmWave) communication greatly improves throughput of wireless networks by using the wide bandwidths available at high frequencies \cite{mmWave-survey-nyu}. In order to establish a viable communication link in highly directional mmWave links and mitigate the high path-loss and intense shadowing, it is necessary to perform beamforming \cite{rappaport2011state}. Beamfomring methods concentrate the transmit power in a desired direction by utilizing narrow beams \cite{kutty2016beamforming}. 

It is known that mmWave channels are sparse and consist of only a few spatial clusters \cite{akdeniz2014millimeter}. Therefore, to reduce beamforming overhead and maximize system throughput, beam alignment (BA) (a.k.a. beam training and beam search) is used to find narrow beams aligned with the direction of the channel clusters \cite{giordani2018tutorial}. In BA, the wireless transceiver searches over the angular space through a set of scanning beams to localize the direction of the channel clusters, i.e., namely, the angle of arrival (AoA) and angle of departure (AoD) of the channel clusters at the receiver and transmitter sides, respectively. Moreover, due to high power consumption in mmWave systems it is often assumed that the transceivers only use one RF-chain during BA, a method known as analog BA. 

There is a large body of work on BA methods in the literature
\cite{khalili2020optimal,michelusi2018optimal,desai2014initial,hussain2017throughput,shah-pimrc,khosravi2019efficient,chiu2019active,shabara2018linear,klautau20185g,Song2019,Lalitha2018}. In general, BA strategies can be classified as \textit{interactive} BA and \textit{non-interactive} BA. To elaborate, let us consider the BA procedure at the transmitter whose objective is to localize the AoD of the channel. The transmitter sends a set of BA packets through a set of scanning beams to scan the angular space. In non-interactive BA, the transmitter does not receive any feedback from the receiver until all the BA  packets are sent. In interactive BA, however, the transmitter receives feedback during the transmission of the BA packets which allows it to refine the scanning beams and better localize the AoD of the channel compared to non-interactive BA. 

The problem of multi-user non-interactive BA is considered in \cite{khalili2020optimal} where we analyzed the BA problem through an information theoretic perspective and provided bounds on the minimum average  expected beamwidth of data beams allocated to the users given a fixed BA duration along with achievablity schemes. A more challenging problem is to consider the interactive case which necessitates optimally utilizing the feedback information during the BA. Prior research on interactive BA methods \cite{michelusi2018optimal,barati2016initial,giordani2016comparative,desai2014initial,hussain2017throughput,shah-pimrc,Shah1906:Robust,khosravi2019efficient,chiu2019active,shabara2018linear,klautau20185g,Song2019} consider no delay for the receiver's feedback on the scanning packets. However, this might not always be the case due to practical reasons such as processing delay at the transceivers. 

In this paper, we consider the problem of interactive analog BA at the base station (BS) in a single-user downlink system where the channel has one dominant cluster and the feedback to each transmitted BA packet is received after a fixed known delay. Due to practical constraints, we only look at the case where the beams are contiguous \cite{khalili2020optimal,michelusi2018optimal,desai2014initial}. Similar to \cite{khalili2020optimal}, we assume that the BA packets and feedback at the user and BS are received error free. As a result, at the end of the BA phase the BS can allocate a beam for the data communication which includes the AoD of the channel with probability one. We refer to the angular region of this beam as \textit{uncertainty region} (UR) on the channel AoD. Our objective is to minimize the expected width of the UR similar to \cite{khalili2020optimal}. 
Overview of the contributions of this paper is as follows:
\begin{itemize}
    \item We view the BA with feedback delay as a source coding problem in which the BS needs to ask $b$ yes/no questions where each question is an angular interval. We show that the resulting source codewords (feedback sequences) have a special characteristic using which we define a new family of codes that we name $d-$unimodal (Section~\ref{sec:BU}).
    
    \item We provide a lower bound on the minimum expected width of the UR given any arbitrary prior on the AoD by analyzing properties of $d-$unimodal codes. Through numerical evaluations, we show the potential improvement in terms of the number of required time-slots to achieve a certain angular resolution for the expected width of the UR when compared with state-of-the-art (Section~\ref{sec:LB}). 
\end{itemize}

\section{System Model and Preliminaries} 
\label{sec:sys}
In this section, we outline general system assumptions (\ref{sec:sys1} and \ref{sec:sys2}) and then provide the mathematical formulation of the problem (\ref{subsec:prb} and \ref{sec:sys4}).

\subsection{Network Model}
\label{sec:sys1}
We consider a single-user downlink communication in a single-cell mmWave system scenario. Motivated by previous works \cite{hussain2017throughput,Shah1906:Robust, shah-pimrc,michelusi-efficient} and experimental results \cite{akdeniz2014millimeter}, we assume that the channel has only a single dominant cluster. 
We denote the AoD corresponding with this cluster by $\psi$ which is unknown to the BS. In our setup, motivated by \cite{chiu2019active,michelusi2018optimal,khalili2020optimal}, we consider that the BS performs analog BA during which it is able to search one beam at a time while the user has an omnidirectional reception pattern. The goal is to find a small angular interval (i.e., UR) which contains $\psi$. We assume $\Psi \sim f_{\Psi} (\psi)$ for $\psi \in(0,2\pi]$ which accounts for the prior knowledge on the AoD (e.g., corresponding to the history of previously localized AoD in beam tracking applications).

Due to practical constraints, we only consider use of contiguous beams as in \cite{khalili2020optimal}. Similar to \cite{chiu2019active,michelusi2018optimal,khalili2020optimal}, we assume that the beams are ideal and use the \textit{sectored antenna} model from \cite{ramanathan2001performance}. In this model, each beam is characterized by a constant main-lobe gain and an angular coverage region (ACR). In the case of contiguous beams, this ACR is an angular interval inside $(0,2\pi]$ that is covered by the main-lobe. This model is often used in the literature (e.g. \cite{bai2015coverage,fund2017}) and is justified as the BSs are envisioned to use large antenna arrays which allows for beams close to ideal shape \cite{mmWave-survey-nyu}. 

\begin{figure}[t]
\centering
\includegraphics[width=0.7\linewidth, draft=false]{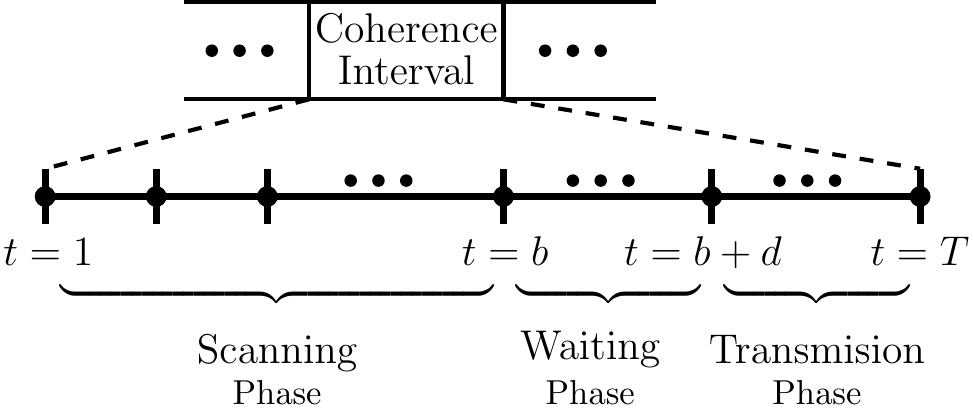}
\caption{Time-slotted system.}
\label{fig:sys_model}
\end{figure}

\subsection{Frames and Feedback}
\label{sec:sys2}

We consider an interactive BA scenario in which the BS receives feedback form the user during the transmission of BA packets and can change the subsequent scanning beams based on the feedback. Unlike conventional interactive BA in which the feedback to each transmitted packet is available instantaneously, we consider a fixed known delay of $d$ time-slots for each feedback. This delay accounts for practical constraints such as processing delay at the transceivers. If this delay cannot be accurately measured, an upper bound can be utilized for our analysis. We assume that the feedback to each packet is either an acknowledgement (ACK) that the packet was received by the user or a negative acknowledgement (NACK) which indicates the user did not receive the packet. Similar to \cite{michelusi2018optimal}, we consider that the feedback is received through a control channel and is error free \cite{mmWave-survey-nyu}. Also, as in  \cite{hussain2017throughput,khalili2020optimal}, we assume that the BA packets are detected at the user without error.

Motivated by the above discussion, we consider a time-slotted system in which the user has fixed AoD over coherence intervals of duration $T$ time-slots. We assume that the communication spanning a coherence interval includes three phases as shown in Fig.~\ref{fig:sys_model}. We first have the \textit{scanning phase} in which the BS transmits $b$ BA packets through a set of scanning beams to scan the angular space. 
Since the response to each packet takes $d$ time-slots, we consider a \textit{waiting phase} in which the BS waits to receive the feedback to the scanning beams. This phase lasts for $d$ time-slots and can be used for example, for data transmission to other users for which the BS has already performed BA. The rest of the coherence interval, i.e., the last $T-b-d$ time slots, is called \textit{transmission phase} which is used for data transmission.

Our main focus in this paper is the design of the beams to be used in the scanning phase and the resulting expected beamwidth for data transmission beam.

\subsection{Scanning Beam Set and Data Beam}
\label{subsec:prb}
The objective of BA is to maximize the beamforming gain which in turn maximizes the data communication rate. Towards this goal, we consider minimizing the expected width of the UR for the AoD of the user's channel. 

The BS uses $b$ scanning beams $\{\Phi_i\}_{i\in[b]}$ to transmit the $b$ BA packets \footnote{We use the notation $[n]$ to represent the set $\{1,2,\ldots,n\}$.}. Let $a_i \in \{0,1\}$ denote the feedback received for the $i^{\rm th}$ BA packet (i.e., $a_i = 1$ if ACK was received for $\Phi_i$ and $a_i = 0$ otherwise). Based on the received feedback sequence by the $i^{\rm th}$ time-slot (i.e., $(a_1,a_2, \ldots, a_{i-d})$), there are multiple choices for $\Phi_i$. To model this, we use a hierarchical beam set $\Sset = \{\Sset_i\}_{i \in [b]}$, where $\Sset_i = \{S_{i,m}\}_{m\in[M(i)]}$ denotes the set of all possible scanning beams given that there are a total of $M(i)\leq 2^{i-d}$ possible feedback sequences. To elaborate, note that by the $i^{\rm th}$ time-slot, the BS has received the feedback sequence corresponding to the scanning beams $\Phi_j, j\leq i-d$. We design the set $\Sset_i = \{S_{i,m}\}_{m\in[M(i)]}$ to contain a beam for each of the possible feedback sequences. Therefore, upon reception of a particular feedback sequence at the $i^{\rm th}$ time-slot, the BS selects the beam $S_{i,m}$ which was designed for that feedback sequence and uses it for transmission of the $i^{\rm th}$ BA packet (i.e., $\Phi_i = S_{i,m}$).

Given an AoD realization $\psi$, let us denote the ACR that the BS chooses for data transmission (i.e., UR) by $\Beam(\Sset,\psi)$. Under the assumption of single dominant path channel and error free system, the minimum length ACR which includes the user AoD is 
\begin{align}
\label{eq:phis}
\Beam(\Sset,\psi) = \cap_{i = 1}^b \Theta(\Phi_i,a_i),
\end{align}
where $\Theta(\Phi_i,a_i) = \Phi_i$ if $a_i = 1$ which corresponds to $\psi \in \Phi_i$, and $\Theta(\Phi_i, a_i) = (0,2\pi] - \Phi_i$ otherwise.

\subsection{Problem Formulation}
\label{sec:sys4}
We formulate the problem of minimizing the expected width of the UR for the AoD as
\begin{align}
\label{eq:optimization} 
\begin{aligned}
\Sset^* = \argmin_{\Sset} \mathbb{E}_{\Psi}[|\Beam(\Sset,\Psi)|],
\end{aligned}
\end{align}
where expectation is taken over the distribution $f_{\Psi}(\psi)$.

Based on \eqref{eq:phis}, given $\Sset$, we get an UR for each possible feedback sequence $(a_1, a_2, \ldots, a_b)$. Let us denote the set of possible URs for the AoD of the user by $\Uset = \{u_m\}_{m \in M(b)}$, where $M(b)\leq 2^b$ is the number of possible feedback sequences. It is easy to see that $\Beam(\Sset,\Psi)  = u_m$ for $\Psi \in u_m$. Hence, we can write the expectation in \eqref{eq:optimization} as 
\begin{align}
\label{eq:optimization2}
\mathbb{E}_{\Psi}[|\Beam(\Sset,\Psi)|] = \sum_{m=1}^{M(b)} |u_m|\int_{\psi\in u_m}\!\!\!\!f_{\Psi}(\psi) d\psi.
\end{align}
The notation $|u_m|$ is the Lebesgue measure of $u_m$, which is equal to the total width of the intervals in the case where $u_m$ is the union of a finite number of intervals. 

In the next sections, we will show that the feedback sequences can be viewed as source codewords with a special characteristic. This characteristic lets us define a new class of codes which we refer to by $d-$unimodal codes. Then, by studying these codes, we lower bound the minimum expected beamwidth in Sec.~\ref{sec:LB}. Explicit construction of optimal scanning beam set is reported elsewhere due to space constraint. 

\section{Beam Alignment and Unimodal Codes}
\label{sec:BU}

We view the discussed BA problem as a source coding problem in which the BS asks $b$ questions whose answers (the feedback sequences) represent the source codewords. Unlike a finite alphabet source coding problem, here the alphabet is continuous and the questions are intervals inside $(0,2\pi]$. 
In this section, we examine the properties of the aforementioned source code in our BA problem and define a new class of codes called \textit{$d-$unimodal}. 
We also establish the connection between the BA schemes and the design of $d-$unimodal codes.

To define $d-$unimodal codes, we need the following 

\begin{Definition}[\textbf{Unimodal Loop}]
A binary loop is called \emph{unimodal} iff the location of ones (if any) are consecutive \footnote{A \emph{loop} is a cyclically ordered set of elements \cite{novak1982cyclically} (i.e., the elements can be arranged on a circle). }.
\end{Definition}

As an example, the loop $\odot\{1,0,0,1\}$ is unimodal but the loop $\odot\{1,0,1,0\}$ is not \footnote{The notation $\odot\{\ldots\}$ indicates the loop of the ordered set $\{\ldots\}$.}. As we will elaborate later, unimodal loops represent the scanning beams in our BA problem. Now, we can define $d-$unimodal codes as follows:

\begin{Definition}[\textbf{$\mathbf{d-}$unimodal Code}] 
\label{def:uni}
A binary code (collection of codewords) with codewords of length $b$ is called \emph{$d$-unimodal} and is denoted by $\Cset(b,d)$, if there exists an ordered set of its codewords which could also include repetition of some codewords whose associated loop satisfies:
\begin{enumerate}
    \item For $i\leq d$, the loop created by the $i^{\rm th}$ bits of the codewords in the loop is unimodal.
    \item For $i>d$, for each sub-loop of the loop consisting only of codewords with same prefix of length $i-d$, the binary loop of the bits in the $i^{\rm th}$ position is unimodal.
\end{enumerate}
We refer to such loop as \emph{characteristic loop} of the code. The \emph{code cardinality} is the number of codewords in $\Cset$, denoted by $|\Cset|$. For example, the code $\Cset = \{11, 01, 10\}$ is a $2-$unimodal code with a characteristic loop of $ \Lset = \odot\{01, 11, 10\}$. More examples are provided later in the paper.
\end{Definition}

Characteristic loop of a code is not unique and may contain repetition of the codewords. For example, consecutively repeating a codeword in a characteristic loop generates another valid characteristic loop. A \emph{minimal characteristic loop (MCL)} is defined as one which does not contain any consecutive repetitions. Yet, an MCL may still contain repetitions that are not consecutive. For example, consider $\Cset = \{11, 01, 10\}$ with the characteristic loop $ \Lset = \odot\{11, 01, 11, 10\}$ which is minimal but contains repetition. 

As part of our first main result (Thm.~\ref{thm:1t1}),  we show that the feedback sequences in our BA problem form a $d-$unimodal code. Moreover, one can also find a construction that given a $d-$unimodal code, generates a scanning beam set $\Sset$ whose feedback sequences are that code. This second part forms the foundation of our explicit construction of optimal BA schemes and will be pursued elsewhere due to space constraint. 

Before providing the theorem statement, we first provide the necessary definitions and show through a set of examples how a scanning beam relates to a unimodal loop and how a given scanning beam set $\Sset$ leads to $d-$unimodal code. 

Suppose we are given a scanning beam set $\Sset$. The scanning beams inside this set, partition the interval $(0,2\pi]$ into a set of angular intervals which we call \emph{component beams}. We define these component beams as follows. Each scanning beam is an angular interval with two endpoints. After sorting the endpoints of all the scanning beams in $\Sset$ and removing the repetitions, each angular interval in between two consecutive endpoints is a component beam. Since the component beams are contiguous and partition the interval $(0,2\pi]$, one can use their positions on the circle and form a loop of the component beams. We denote this loop using $\Iset$. 
To better understand the notation and the relation between $\Sset$ and $\Iset$, let us consider the following example 
which we will build upon in the paper.

\begin{figure*}[t!]
    \centering
    \includegraphics[width = 0.7\textwidth]{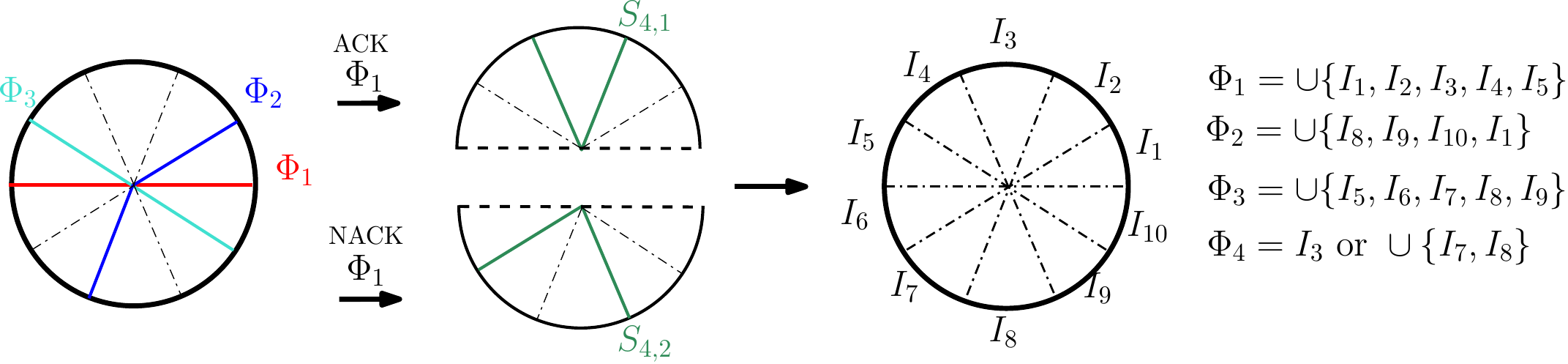}
     \caption{An example set of scanning beams for $b=4$ and $d=3$ and the corresponding component beams.}
    \label{fig:SCLI}
\end{figure*}

\begin{Example}
\label{ex:Iset}
Fig.~\ref{fig:SCLI} illustrates a possible set of scanning beams for $b=4$ and $d=3$. In this case, $\Sset = \{\Sset_1,\Sset_2,\Sset_3,\Sset_4\}$, with $\Sset_i = \{\Phi_i\}$ for $i \in \{1,2,3\}$ each consisting of a single possible scanning beam as no feedback is received prior to fourth time-slot. However, at the fourth time-slot, we receive the feedback to the beam $\Phi_1$ and so there are two possibilities for $\Phi_4$. Here,  we have $\Sset_4 = \{S_{4,1}, S_{4,2}\}$. As shown in Fig.~\ref{fig:SCLI}, the set $\Sset$ creates the component beam loop $\Iset = \odot\{I_1,I_2,\ldots,I_{10}\}$ which includes ten component beams.
\end{Example}

It is easy to see that each of the scanning beams in $\Sset$ can be written as union of subsets of component beams in $\Iset$. Consider one of these beams. By replacing the elements of the loop $\Iset$ with $1$ if the component beam is included in the beam and $0$ otherwise, we will have a binary loop. As a result, we can uniquely determine any beam in $\Sset$ using a binary loop given $\Iset$. Note that since the scanning beams are contiguous, they each include adjacent component beams and so the position of ones inside their binary loops are consecutive. Therefore, these binary loops are unimodal. 
To elaborate, consider the setup of Example~\ref{ex:Iset}. The beam $S_{4,2}$ in Fig.~\ref{fig:SCLI} is partitioned by component beams $I_7$ and $I_8$. Hence, its corresponding binary loop is $\odot\{0,0,0,0,0,0,1,1,0,0\}$ which is unimodal. 

Next, we show how these unimodal binary loops corresponding to the scanning beams lead to $d-$unimodal codes. For this purpose, let us first consider the following example. 

\begin{Example}
\label{ex:du}
Consider the setup in Example~\ref{ex:Iset}. If we replace the component beams in the loop $\Iset$ with their corresponding feedback sequences, we get the loop $\Lset= $  $\odot \{1100,$ $1000,$ $1001,$ $1000,$  $1010,$ $0010,$ $0011,$ $0111,$ $0110,$ $0100\}$. Let us look at the second bit of the codewords in $\Lset$ which leads to the binary loop $\odot \{1,$ $0,$ $0,$ $0,$ $0,$ $0,$ $0,$ $1,$ $1,$ $1\}$. This loop is unimodal and is the same binary loop representing the contiguous scanning beam $\Phi_2  =$ $\cup\{I_8,$ $I_9,$ $I_{10},$ $I_1\}$. Next, consider the loop of the fourth bits $\odot \{0,$ $0,$ $1,$ $0,$ $0,$ $0,$ $1,$ $1,$ $0,$ $0\}$. This is not a unimodal loop. Here, this loop corresponds to the beam $S_{4,1} \cup S_{4,2}$ which is not contiguous. Notice that each of these beams are contiguous. So if we can separate their binary loops, we should get unimodal loops. Note that the decision that which one of these beams is used for $\Phi_4$ is based on the feedback sequences received at the fourth time-slot which is $a_1$ in our case. Now, if we look at the sub-loops of the loop $\Lset$ whose feedback sequences have same value for $a_1$, we will get the binary loops $\odot \{0,$ $0,$ $1,$ $0,$ $0\}$ and  $\odot\{0,$ $1,$ $1,$ $0,$ $0\}$\footnote{A sub-loop of a loop is a loop in which some of the elements of the original loop are removed.}. These are both unimodal loops where the former and the latter correspond to  $S_{4,1}$ and $S_{4,2}$, respectively which are contiguous.
\end{Example}

Let us form a loop of binary codewords by replacing each component beam in loop $\Iset$ with its corresponding feedback sequence (i.e., $(a_1,a_2,\ldots, a_b)$) and denote it by $\Lset$. Following Example \ref{ex:du}, if we look at the loop consisting of the $i^{\rm th}$ bit of each codeword, we have a binary loop which relates to the scanning beams in $\Sset_i$. For $i\leq d$, since we have not received any feedback, $\Sset_i$ consists of only one contiguous scanning beam and the binary loop becomes unimodal. However, for $i>d$, this is no longer the case since there are multiple scanning beams in $\Sset_i$. Yet, similar to Example~\ref{ex:du}, if we create sub-loops of $\Lset$ that include feedback sequences of same prefix of $(a_1, a_2,\ldots, a_{i-d})$ and then look at the loop of the $i^{\rm th}$ bits for each sub-loop, we will get unimodal loops.
These claims are proved rigorously in the next theorem were we show that the loop $\Lset$ is in fact an MCL for the $d-$unimodal code whose codewords are the feedback sequences resulting from $\Sset$.

An interesting observation from the MCL created using the feedback sequences and the component beams loop is that when it has repetition, one or more of the URs are non-contiguous. The repetition of a codeword means there are multiple component beams with same feedback sequence and adjacent component beams of different feedback sequences. on the other hand, from Sec.~\ref{sec:sys4}, we know that each feedback sequence corresponds to an UR. Therefore, there is an UR that includes these component beams but not their adjacent which makes it non-contiguous. This is important since as we discussed in Sec.~\ref{sec:sys}, each UR is a possible data beam and the data beams are preferred to be contiguous. As an example of this observation, consider Example~\ref{ex:du}. The MCL has the repetition of the codeword $1000$. On the other hand, if we form the set of URs, we get $\Uset = \{u_m\}_{m=1}^9$, where $u_m = I_m$ for $m\in [10]\backslash2$ and $u_2 = \cup\{I_2,I_4\}$. The non-contiguous UR is $u_2$ whose feedback sequence is the repeated codeword $1000$.

\begin{Theorem}[\textbf{Beam Alignment and Unimodal Codes}]
\label{thm:1t1}
Consider the BA problem introduced in Section \ref{sec:sys} where the number of BA scanning packets is set to $b$ and the delay is set to $d$. Given any scanning beam set $\Sset$, the feedback sequences form a $d-$unimodal code $\Cset$ whose MCL $\Lset$ is the loop of binary codewords resulted from replacing the elements of the component beams $\Iset$ with their corresponding feedback sequences.
\end{Theorem}
\begin{proof}
The proof is provided in Appendix A.
\end{proof}
This theorem shows that the collection of feedback sequences of any possible scanning beam set is a $d-$unimodal code. We will use this to lower bound the performance of the considered BA problem in terms of minimum expected beamwidth in the next section. 

\section{Lower Bound on Expected Beamwidth}
\label{sec:LB}
In this section, we investigate the properties of $d-$unimodal codes to lower bound the optimal performance in terms of expected beamwidth for our BA problem. To this end, 
we define a parent-child hierarchy between the codes $\Cset(b,d)$ and $\Cset(b-1,d)$ which we will use in our proofs. This hierarchy is formally defined below.

\begin{Definition}[\textbf{Parent Code}]
For a $\Cset(b,d)$ code with an MCL $\Lset(b,d)$, the loop containing the prefix of length $b-i$ of all the codewords in the loop is an MCL that defines a \emph{parent code of order $i$}, i.e.,  $\Cset(b-i,d)$. The parent code of order $1$ is simply called the \emph{parent code}. 
\end{Definition}
It can be inferred that given a code, its corresponding parent code is unique  and $d$-unimodal. However, a parent code can result in different child codes. Note that based on Thm.~\ref{thm:1t1}, given a scanning beam set, the resulting collection of feedback sequences is a $d-$unimodal code. Also, from Sec.~\ref{sec:sys4}, we know that the number of possible URs is the same as the number of possible feedback sequences. As a result, we can upper bound the number of URs (number of feedback sequences) by finding an upper bound for the cardinality of $d-$unimodal codes. 
In the next theorem, we use the parent-child hierarchy to bound the cardinality of $d-$unimodal codes which also gives us a bound on the number of URs. 

\begin{Theorem}[\textbf{Maximum Code Cardinality}]
\label{thm:opmax}
Let $M(b,d)$ denote the maximum cardinality for the code ${\Cset}(b,d)$. Then,
for $d =1$, $M(b,d) = 2^b$ and for $d>1$,
\begin{align}
\label{eq:maxd}
\begin{aligned}
M(b,d) &\leq \begin{cases}
M(b-1,d) + 2M(b-d,d) &b>d, \\
2b  &b \leq d.
\end{cases}
\end{aligned}
\end{align}
\end{Theorem}
\begin{proof}
The proof is provided in Appendix B.
A sketch of which is as follows. From Def.~\ref{def:uni}, we know that the cardinality of a $d-$unimodal code is less than or equal to the length of its MCL. Moreover, it is easy to see that the length of any MCL for $\Cset(1,d)$ is at most $2$. On the other hand, using parent-child hierarchy, we bound the difference between the cardinality of a child's MCL with its parent's MCL. Based on these, we calculate the upper bound in the theorem.
\end{proof}

Using the above results, we can bound the minimum expected beamwidth for UR as in the next theorem. 

\begin{Theorem}[\textbf{Minimum Expected Beamwidth}]
\label{thm:bounds}
The minimum expected beamwidth i.e., the objective function in optimization problem \eqref{eq:optimization} when contiguous scanning beams are used is bounded as 
\begin{align}
\label{eq:con}
\frac{2^{h(\Psi)}}{M(b,d)} \leq \mathbb{E}_{\Psi}[|\Beam(\Sset,\Psi)|]
\end{align}
\end{Theorem}
\begin{proof}
From Thm.~\ref{thm:1t1} and Thm.~\ref{thm:opmax}, we observe that the maximum number of URs is bounded by $M(b,d)$. Using this with \cite[Prop. 2]{khalili2020optimal}, will give us the lower bound.
\end{proof}

\begin{figure}[t!]
    \centering \includegraphics[width = 0.75\linewidth]{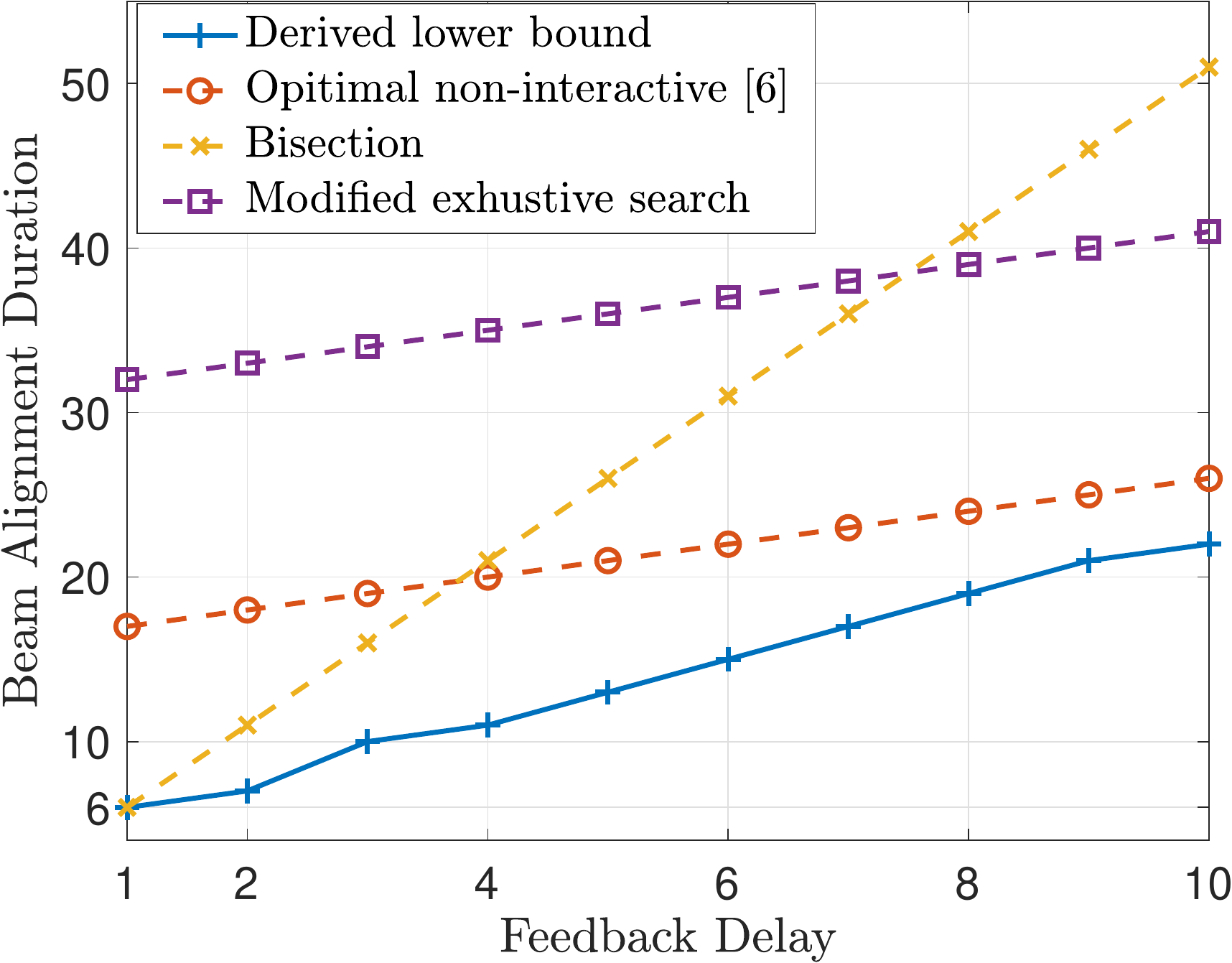} 
\caption{ Total BA duration for a given fixed expected data beamwidth resolution of $360/2^5 \approx 10^{\circ}$ for different BA methods and feedback delays for $\psi \sim \mathrm{Uniform}(0,2\pi]$.}
\label{fig:comp_naive}
\end{figure}

We conclude this section, by providing a comparison of the total (i.e., scanning phase $+$ waiting phase) BA duration that different BA methods and the derived lower bound require, given a fixed expected UR width for different values of feedback delay and $\Psi\sim \mathrm{Uniform}(0,2\pi]$. The result is plotted in Fig.~\ref{fig:comp_naive}. In the modified exhaustive search method, for a given $b$, we divide the $(0,2\pi]$ into $b+1$ equal width URs, and at each time-slot, scan one UR. Since the system is error free, we can find the UR including the AoD by only searching $b$ of $b+1$ URs. We observe that as the delay increases, the performance of bisection method which is optimal for case of $d=1$ rapidly degrades and after delay of $ d = 8$ time-slots, even the modified exhaustive search outperforms bisection method. This figure also shows that as the delay increases the lower bound becomes closer to the performance of the optimal non-interactive method \cite{khalili2020optimal}. In fact, if we allow for more delays, they become exactly the same. The reason is that the optimal non-interactive method in \cite{khalili2020optimal} is a special case of our problem for $d>b$. This plot also suggests that there is potential of improving the performance of the state-of-the-art methods using an appropriate BA scheme. In fact, the proposed framework can also be used to construct the optimal BA method achieving the lower bound in the Fig.~\ref{fig:comp_naive}. 
Details and the derivation of optimal BA solution are left for future publication due to space constraint.

\section{Conclusion}
In this paper, we have investigated the single-user analog BA, where there is a fixed delay between each transmitted BA packet and its corresponding feedback. We have proposed a general framework for this problem using $d-$unimodal codes. We have shown that the feedback sequences form a class of codes we refer to by $d-$unimodal codes, using which we have derived a lower bound on the minimum feasible expected width of the URs. Furthermore, through numerical evaluation, we have shown the possibility of performance improvement over the state-of-the art methods in terms of BA duration required to achieve a fixed expected width for the UR. 

\bibliographystyle{IEEEtran}
\bibliography{bibl}

\appendices
\section{Proof of Theorem \ref{thm:1t1}}
\label{app:1t1}
Let us denote the loop resulting from replacing each component beam in $\Iset$ with its corresponding feedback sequences by $\Lset$.
To prove the theorem statement, we show that $\Lset$ firstly, does not have any consecutive repetitions and secondly, is a characteristic loop for a code consisting of all the feedback sequences of $\Sset$. We prove the second part by showing $\Lset$ satisfies each of the conditions in Def.~\ref{def:uni}. 

\textit{No consecutive repetitions:} We show this by contradiction. Assume $\Lset$ has a consecutive repetition. This means two consecutive component beams lets say $I_1$ and $I_2$ have same feedback sequences. If so, all the scanning beams in $\Sset$ either include both $I_1$ and $I_2$ or none. 
Therefore, if we form the component beams of $\Sset$, we should get $I_1 \cup I_2$ as a component beam instead of two separate component beams $I_1$ and $I_2$. This is a contradiction and so $\Lset$ does not have any repetitions.

To prove that $\Lset$ is a characteristic loop of a code consisting of all the feedback sequences. First, note that the loop $\Lset$ includes all the possible feedback sequences as the component beams of $\Iset$ partition the entire $(0, 2\pi]$. Next, we show that $\Lset$ satisfies the conditions in Def.~\ref{def:uni}. 

\textit{Def.~\ref{def:uni}, first condition:} Assume $i\leq d$, since no feedback is received, the set $\Sset_i$ includes only one scanning beam. Let us denote this beam by $\Phi_i$. By construction of component beams loop $\Iset$, $\Phi_i$ can be written as union of subset of component beams in $\Iset$. It is easy to see that if a component beam is included in $\Phi_i$, the feedback sequence corresponding to that component beam would have one in the $i^{\rm th}$ position and otherwise zero. Now, Let us consider the binary loop derived from the $i^{\rm th}$ elements of the feedback sequences in $\Lset$ and denote it by $\Lset_i$. Based on our discussions, in $\Lset_i$, at the position of the component beams which are included in $\Phi_i$, we should have one and otherwise we should have zero. Since $\Phi_i$ is contiguous, the component beams included in $\Phi_i$ are adjacent and so the positions of ones in $\Lset_i$ are consecutive. Therefore, $\Lset_i$ is by definition unimodal and the condition holds.

\textit{Def.~\ref{def:uni}, second condition:} Suppose that we are at the $i^{\rm th}$ time-slot of BA where $i\geq d$. Without loss of generality assume that the BS has received the feedback sequence $(a_1, a_2, \ldots a_{i-d})$ for which it uses the scanning beam $S_{i,m}$ from $\Sset_i$ for some $m\in[M(i)]$. We need to show that if we form the sub-loop of the feedback sequences in $\Lset$ that have the prefix $(a_1, a_2, \ldots, a_{i-d})$, the binary loop of the $i^{\rm th}$ bits is unimodal. Let us denote this binary loop by $\tilde{\Lset}_i$.

Note that the scanning beams $\Phi_j, j \leq i-d$ and the considered feedback sequence $(a_1,a_2, \ldots,a_{i-d})$ determine an angular region that includes the AoD of the user at the $i^{\rm th}$ time-slot. By construction of the component beam loop $\Iset$, we can write this angular region as a union of subset of the component beams in $\Iset$. Given this, let us form a sub-loop of $\Iset$ by removing the component beams which are not included in this angular region and name it $\tilde{\Iset}$. Consider one of the component beams inside this sub-loop and denote it by $\tilde{I}_1$. It is straight forward to see that the $i^{\rm th}$ bit of the feedback sequence $(a_1,a_2, \ldots,a_{b})$ corresponding to $\tilde{I}_1$ is one if  $ \tilde{I}_1\in S_{i,m}$ and zero otherwise. Therefore, if we replace the component beams inside $\tilde{\Iset}$ with their feedback sequences and look at the loop of $i^{\rm th}$ bits it would be unimodal since $S_{i,m}$ is contiguous. Next, we show that this loop is the loop $\tilde{\Lset}_i$. To show this, note that $\tilde{\Iset}$ is a sub-loop of $\Iset$ and includes all and only component beams of $\Iset$ whose feedback sequences have the prefix $(a_1, a_2, \ldots, a_{i-d})$. Therefore, if we replace the component beams inside $\tilde{\Iset}$ with their corresponding feedback sequences and form the binary loop of the $i^{\rm th}$ bits by definition it will be the loop $\tilde{\Lset}_i$.

\section{Proof of Theorem~\ref{thm:opmax}}
\label{app:opmax}
For the case of $d=1$, we know that the maximum number of feedback sequences using $b$ yes/no questions is $2^b$ which is also achievable (e.g., bisection method). Therefore, $M(b,1) = 2^b$.

For $d>1$, we first bound the cardinality of the MCL $\Lset(b,d)$, and then bound $M(b,d)$ using Def.~\ref{def:uni} which indicates $M(b,d) \leq  |\Lset(b,d)|$.

Looking at the Def.~\ref{def:uni}, by grouping the codewords in the MCL $\Lset(b,d)$ into sub-loops whose codewords have same prefix of length $b-d$, the loop of the last bit of the codewords in each sub-loop becomes unimodal. On the other hand, we know that if we remove the last bit of all the codewords in the MCL $\Lset(b,d)$ and eliminate the consecutive repetitions of the created codewords, by definition of the parent code, we will get an MCL $\Lset(b-1,d)$. Note that if no consecutive repetitions are caused by removing the last bit of the codewords, we would have $|\Lset(b,d)| = |\Lset(b-1,d)|$. However, if there were codewords to be eliminated due to the consecutive repetitions, we would get $|\Lset(b,d)| > |\Lset(b-1,d)|$. 
So, by finding the maximum possible reduction of codewords due to the consecutive repetitions, we can find the maximum possible difference between $|\Lset(b,d)|$ and $|\Lset(b-1,d)|$. To count this, observe that the MCL $\Lset(b,d)$ cannot have any consecutive repetitions by definition. Therefore, if the bits removed from two consecutive codewords in $\Lset(b,d)$ are the same, they cannot lead to consecutive repetitions in $\Lset(b-1,d)$. Also, removing the last bit of the consecutive codewords in $\Lset(b,d)$ which have different prefixes of length $b-d$ does not lead in consecutive repetitions either since the codewords don't have the same first $b-d$ bits.
As a result, the only way that consecutive repetitions might happen is when the last bits are not the same and the codewords have same prefix of length $b-d$. Suppose we group the codewords into sub-loops that have same prefix of length $b-d$. We know that the loop of the last bits in each sub-loop is unimodal. Moreover, the maximum number of times that two consecutive bits in a unimodal loop can be different is $2$. Therefore, 
\begin{align}
\label{eq:mclb1}
 |\Lset(b,d)| \leq  |\Lset(b-1,d)| + 2(\textrm{number of sub-loops}). 
\end{align}
The number of sub-loops of codewords with same prefix of length $b-d$ for any possible MCL of $C(b,d)$ is by definition of a parent code equal to $|C(b-d,d)|$. Combining this with \eqref{eq:mclb1} and using the fact that $M(b-d,d) \leq |\Lset(b-d,d)|$, we get 
\begin{align}
\label{eq:mclb2}
 M(b,d) \leq |\Lset(b,d)| \leq  |\Lset(b-1,d)| + 2 |\Lset(b-d,d)|. 
\end{align}
When $b\leq d$, we can conclude form in \cite[Prop.~6]{khalili2020optimal} that $M(b,d) = |\Lset(b,d)| = 2b$. Therefore, 
\begin{align}
\label{eq:mclb3}
 M(b,d) = |\Lset(b,d)| = 2b. \quad \textrm{for } b\leq d
\end{align}
Combining this with \eqref{eq:mclb2}, we get the bound in the theorem which concludes the proof.
\end{document}

%% file: commands.tex
\usepackage[normalem]{ulem} 

\newcommand{\citep}[1]{\cite{#1}}

\renewcommand{\tilde}{\widetilde}

\def\beq{\begin{equation}}
\def\eeq{\end{equation}}
\def\beqa{\begin{eqnarray}}
\def\eeqa{\end{eqnarray}}
\def\beqan{\begin{eqnarray*}}
\def\eeqan{\end{eqnarray*}}

\DeclareMathOperator*{\argmin}{arg\,min}

\def\Cset{{\cal C}}
\def\Lset{{\cal L}}
\def\Uset{{\cal U}}
\def\Iset{{\cal I}}
\def\Sset{{\cal S}}

\newcommand{\Beam}{\text{\sf Beam}}

